\documentclass[]{interact}

\usepackage{epstopdf}
\usepackage{subfigure}
\usepackage{array, multirow}
\usepackage{float}
\usepackage{blindtext}
\usepackage[numbers,sort&compress]{natbib}
\bibpunct[, ]{[}{]}{,}{n}{,}{,}
\newcommand{\f}{\operatorname}

\theoremstyle{plain}
\newtheorem{theorem}{Theorem}[section]

\newtheorem{proposition}[theorem]{Proposition}

\theoremstyle{definition}

\theoremstyle{remark}

\begin{document}

\articletype{ }

\title{An Extended Poisson Family of Life Distribution: A Unified Approach in Competitive and Complementary Risks}

\author{ Pedro L. Ramos$^{\rm ab}$$^{\ast}$\thanks{$^\ast$Corresponding author. Email: pedrolramos@usp.br
\vspace{6pt}}, Dipak K. Dey$^{\rm b}$, Francisco Louzada$^{\rm a}$ and Victor H. Lachos$^{\rm b}$ \\\vspace{6pt}  $^{a}${Institute of Mathematical Science and Computing, University of  S\~ao \\ Paulo, S\~ao Carlos, Brazil} \\ 
$^{b}${Department of Statistics, University of  Connecticut, Storrs, CT, USA}}

\maketitle

\begin{abstract}
In this paper, we introduce a new approach to generate flexible parametric families of distributions. These
models arise on competitive and complementary risks scenario, in which the lifetime associated
with a particular risk is not observable; rather, we observe only the minimum/maximum lifetime
value among all risks. The latent variables have a zero truncated Poisson distribution. For the proposed family of distribution, the extra shape parameter has an important physical interpretation in the competing and complementary risks scenario. The mathematical properties and inferential procedures are discussed. The proposed approach is applied in some existing distributions in which it is fully illustrated by an important data set.
\end{abstract}


\section{Introduction}

There has been a renewed interest among researchers in presenting new class of distribution for describing these problems. For instance, Marshall and Olkin \cite{marshall1997new} introduced a new procedure for adding new parameters in common distributions. The authors showed that depending on the values of the new parameter, the new distribution may arrive from the minimum or maximum, where the latent variable follows a geometric distribution and each of components in risk came from a baseline distribution. Further, many special cases of this family were considered, see e.g., Barreto-Souza et al. \cite{barreto2013general} and the references therein.

Another distribution that has been considered as latent variable is the zero truncated Poisson (ZTP) distribution. Kus \cite{kus} derived the exponential-Poisson (EP) distribution by taking the minimum among the lifetimes, where the baseline is the exponential distribution and the latent variable has ZTP distribution. Cancho et al. \cite{cancho} follow the opposite way and considered the maximum, the obtained model there is known as Poisson-exponential (PE) distribution. Lu and Shi \cite{lu2012new} presented the Weibull-Poisson distribution as generalization of EP distribution. Barreto-Souza and Cribari-Neto \cite{barreto2009generalization} discussed another generalized exponential-Poisson distribution by inserting a power parameter in EP distribution. In fact, Tahir and Cordeiro \cite{tahir2016compounding} reviewed more than 20 already introduced distributions based on the zero truncated Poisson distribution.

Although the procedures used to generate the new distributions seems to produce models with different forms, we discussed a unified approach to construct these distributions by extending one of the parameters to the negative space. More importantly this unified approach does not include any additional parameter. For instance, Kus \cite{kus} and Cancho et al. \cite{cancho} distributions can be merged into one only by considering the shape parameter into the positive and negative space. This distribution can be defined as extended exponential-Poisson distribution and its shape parameter has an important interpretation in competitive and complementary risk (CCR) scenarios, representing the lifetime of the minimum (maximum) according to the values of the shape parameter (negative or positive). 
Indeed, CCR problems arise in several areas such as, biomedical studies, reliability and demography. In CCR problems the lifetime associated with a particular risk is not observable; rather, we observe only the minimum or the maximum lifetime of all risks \cite{louzada1999}.

We also noted that positive parameter space has been extended into negative space for the inverse Gaussian and Gompetz distribution earlier (see, for instance,  Balka et al. \cite{balka2009review, balka2011bayesian}). Whitmore \cite{whitmore1979inverse} introduces the term ''defective'' to deal with these types of distributions. In this case, the defected distributions have improper survival functions and can be used to model the cure fraction of patients \cite{rocha2017new}. Here, the term defected is avoided since the obtained distributions have proper survival functions. Further, we discuss the same approach for other especial cases.

The remainder of this paper is organized as follows: Section 2 presents the genesis and some mathematical properties of our proposed family of distributions. Section 3 shows the maximum likelihood (ML) estimators in the presence of censorship and its properties. Section 4 presents the application of our proposed approach in some common distribution. Section 5 illustrates the obtained models to fit an airplane lifetime data. Some final comments are made in Section 6.

\section{Genesis and Properties}

In this section, we discuss the genesis and some properties of the proposed family of distributions.

\subsection{Competitive risks}

 Let $Y_i \ (i=1,2,.\ldots)$ denote the time-to-event due to the $j$-th competitive risks and N be a random variable with a zero-truncated Poisson (ZTP) distribution indexed by parameter $\phi$, hereafter ZTP($\phi$), given by
\begin{equation*}
P(N=n)=\frac{\phi^n}{n!(e^\phi-1)}, \quad n\in\mathbb{N} \ \mbox{and} \ \phi>0.
\end{equation*}

Now, let $X=\min\left\{Y_i\right\}_{i=1}^{N}$, where $X_i$ are independent of N and assumed to be independent and identically distributed according to a uniform distribution in the interval (0,1). The conditional cumulative distribution function (cdf) of $X$ is given by
\begin{equation*}
F(x;N)=1-P(X>x;N)=1-\left(P(Y_1>X \right)^N=1-\left(1-x \right)^N .
\end{equation*}

Thus, the unconditional cdf of T is
\begin{equation*}
F(x;\phi)=\sum_{n=1}^{\infty}\frac{\phi^n}{n!(e^\phi-1)}-\sum_{n=1}^{\infty}\frac{\phi^n}{n!(e^\phi-1)}\left(1-(1-t)^n \right)=\frac{1-e^{-\phi x}}{1-e^{-\phi}} \cdot
\end{equation*}

Substituting $0\leq x\leq 1$ by a generic cdf $F(x;\boldsymbol{\theta})$ we have that the pdf is given by
\begin{equation}\label{cdfcompet}
g(x;\boldsymbol{\theta},\phi)=\frac{\phi}{1-e^{-\phi}}f(x;\boldsymbol{\theta})e^{-\phi F(x;\boldsymbol{\theta})} ,
\end{equation}
where $f(x;\boldsymbol{\theta})$ is the baseline distribution and $F(x;\boldsymbol{\theta})$ is the baseline cumulative function.

\subsection{Complementary risks}

Considering the competitive risk scenario, let $Z_i \ (i=1,2,.\ldots)$ denote the time-to-event due to the $j$-th competitive risks and N follows a ZTP($\lambda$) distribution. Now, let $W=\max\left\{Y_i\right\}_{i=1}^{N}$ where $X_i$ are independent of N and assumed to be independent and identically distributed according to an uniform distribution in the interval (0,1). The conditional probability density function (pdf) of $X$ given $N$ is given by
\begin{equation*}
f(w;N=n)=nw^{n-1}, w>0 \ n=1,2,\ldots
\end{equation*}

Thus, the unconditional p.d.f. of T is given by
\begin{equation*}
\begin{aligned}
f(t;\lambda)&=\sum_{n=1}^{\infty}nt^{n-1} \frac{\lambda^ne^{-\lambda}}{n!(1-e^{-\lambda})}=\frac{\lambda e^{-\lambda}}{(1-e^{-\lambda})}\sum_{n=1}^{\infty} \frac{\left(\lambda t\right)^{n-1}}{(n-1)!} =\frac{\lambda e^{-\lambda-\lambda t}}{(1-e^{-\lambda})} \cdot
\end{aligned}
\end{equation*}

Substituting $0\leq t\leq 1$ by a generic cdf $F(t;\boldsymbol{\theta})$ we have
\begin{equation}\label{cdfcomplemen}
g(t;\boldsymbol{\theta},\lambda)=\frac{\lambda}{1-e^{-\lambda}}f(t;\boldsymbol{\theta})e^{-\lambda\bar{F}(t;\boldsymbol{\theta})},
\end{equation}
where $f(t;\boldsymbol{\theta})$ is the pdf related to baseline distribution.

\subsection{A unified approach}

Note that, although (\ref{cdfcompet}) seems to differ from (\ref{cdfcomplemen}), if we assume that $\lambda$ takes negative values, i.e., $\lambda=-\phi$, 
from (\ref{cdfcomplemen}) we have
\begin{equation}\label{ppdfmax}
f(t;\boldsymbol{\theta},\phi)= -\frac{\phi}{1-e^{\phi}}f(t;\boldsymbol{\theta})e^{\phi(1- F(t;\boldsymbol{\theta}))}=\frac{\phi}{1-e^{-\phi}}f(t;\boldsymbol{\theta})e^{-\phi F(t;\boldsymbol{\theta})},
\end{equation}
which is the same family of distributions presented in (\ref{cdfcompet}) but considering $\phi<0$. The relation in (\ref{ppdfmax}) holds since
\begin{equation*}
-\frac{\phi}{1-e^{\phi}}=\frac{\phi}{e^\phi-1}=\frac{\phi}{e^\phi-1}\frac{e^{-\phi}}{e^{-\phi}}=\frac{\phi e^{-\phi}}{1-e^{-\phi}} \cdot
\end{equation*}
Hence, both distributions can be unified in a simple form. Let $\mathbb R^*=\mathbb R/\{0\}$, if X has an extended Poisson-family of distributions then its cumulative distribution
function is given by
\begin{equation}\label{cumpe}
G(t;\boldsymbol{\theta},\lambda)=\frac{e^{-\lambda(1- F(t;\boldsymbol{\theta}))}-e^{-\lambda
}}{1-e^{-\lambda}}=\frac{e^{\lambda F(t;\boldsymbol{\theta})}-1}{e^{\lambda}-1},
\end{equation}
for all $t>0$, where $\lambda\in \mathbb R^*$ is a shape and $\boldsymbol{\theta}$ is a vector of parameters related to the parametric baseline distribution. Additionally, the survival function $\bar{G}(\cdot)$ and the hazard function $h(\cdot)$ are given, respectively, by
\begin{equation}\label{sumpe}
\bar{G}(t;\boldsymbol{\theta},\lambda)=\frac{1-e^{-\lambda \bar{F}(t;\boldsymbol{\theta})}}{1-e^{-\lambda}}, \quad \mbox{and} \quad h(t;\boldsymbol{\theta},\lambda)=\frac{\lambda f(t;\boldsymbol{\theta})}{e^{\lambda \bar{F}(t;\boldsymbol{\theta})}-1} \cdot
\end{equation}

The shape parameter of our class of models has an important interpretation in the competing and complementary risks scenario, e.g., under the above assumptions if $\lambda<0$ ($\lambda>0$) then T represents the lifetime of the minimum (maximum) of $Y_i$. Moreover, as $\lambda$  tends to zero, the new family of distribution converges to the baseline distribution (random). This family of distribution has an interesting property related to the quantile function, if the quantile function of the baseline distribution has closed-form expression the quantile function related to the composed distribution has also closed-form expression.

\begin{proposition}\label{propclo1} Let $F^{-1}(p;\boldsymbol{\theta})$ be the quantile function of $F(t;\boldsymbol{\theta})$, if $F^{-1}(p;\boldsymbol{\theta})$ has closed-form then $G^{-1}(p;\boldsymbol{\theta},\lambda)$ also has closed form and $G^{-1}(p;\boldsymbol{\theta},\lambda) = F^{-1}(\f{q}(p,\lambda);\boldsymbol{\theta})$ where
\begin{equation}\label{qplambda}
\f{q}(p,\lambda) = \frac{\log\left((e^{\lambda} - 1)p + 1\right)}{\lambda}.
\end{equation}
\end{proposition}
\begin{proof} We have that $G(t;\boldsymbol{\theta},\lambda)=p\Leftrightarrow G^{-1}(p;\boldsymbol{\theta},\lambda) = t$ by definition. But
\begin{equation*}
\begin{aligned}
G(t;\boldsymbol{\theta};&\lambda) = p \Leftrightarrow
\frac{e^{\lambda F(t;\boldsymbol{\theta})}-1}{e^{\lambda}-1} = p \Leftrightarrow 
e^{\lambda F(t;\boldsymbol{\theta})} = (e^{\lambda}-1)p + 1 \Leftrightarrow \\& F(t;\boldsymbol{\theta}) = \frac{\log\left((e^{\lambda} - 1)p + 1\right)}{\lambda} \Leftrightarrow t = F^{-1}(\f{q}(p,\lambda);\boldsymbol{\theta}),
\end{aligned}
\end{equation*}
where $\f{q}(p,\lambda)$ is given by (\ref{qplambda}). Therefore $G^{-1}(p;\boldsymbol{\theta},\lambda) = F^{-1}(\f{q}(p,\lambda);\boldsymbol{\theta})$.
\end{proof}

\subsection{Presence of instantaneous failures}

When data to be modeled has the presence of instantaneous failures (inliers), standard distributions may not be suitable. For instance, in lifetime testing of electronic devices the occurrence of fail at time 0 may be observed due to inferior quality or construction problem. Another example is in weather forecasts where the occurrence of dry periods without the presence of precipitation is very common, standard models such as Gamma, Weibull, Lognormal cannot be used. Although for our family of distribution the parametric baseline distributions are greater than zero, some of the compound models may allow the occurrence of zero value. 
\begin{proposition}
Let $h_F(t;\boldsymbol{\theta})$ be the baseline hazard function related to the cdf $F(t;\boldsymbol{\theta})$, then if $0<h_F(0;\boldsymbol{\theta})<\infty$, we have $g(0;\boldsymbol{\theta},\lambda)> 0$.
\end{proposition}
\begin{proof}
Note that
\begin{equation}\label{zeroprof}
\begin{aligned}
\lim_{t \to 0} g(t;\boldsymbol{\theta},\lambda)&= \lim_{t \to 0} \frac{\lambda}{1-e^{-\lambda}}f(t;\boldsymbol{\theta})e^{-\lambda(1- F(t;\boldsymbol{\theta}))}\\ & = \lim_{t \to 0} \frac{\lambda}{1-e^{-\lambda}}h_F(t;\boldsymbol{\theta})\bar{F}(t;\boldsymbol{\theta})e^{-\lambda(\bar{F}(t;\boldsymbol{\theta}))}
\\ & = \frac{\lambda}{e^{\lambda}-1}h_F(0;\boldsymbol{\theta}).
\end{aligned}
\end{equation}

Since $\lambda(e^{\lambda}-1)^{-1}>0, \forall\, \lambda\in\mathbb R^*$, then for $0<h_F(0;\boldsymbol{\theta})<\infty$,  we have  $g(0;\boldsymbol{\theta},\lambda)> 0$.

\end{proof}

Therefore, depending on the behavior of the baseline hazard function, some of the resulting models will be capable to accommodate data with zero value.

\section{Inference}

In statistical inference, different procedures can be considered in order to obtain the parameter estimates of particular distributions \cite{rodrigues2016poisson}. The ML estimators are usually considered due to its attractive limiting properties such as consistency, asymptotic normality and efficiency \cite{migon2014}. 

Let $T_i$ be the lifetime of $i$th  component with censoring time $C_i$, which are assumed to be independent of
$T_i$s and its distribution does not depend on the parameters, the data set is represented by $\mathcal{D}=(t_i,\delta_i)$, where $t_i=\min(T_i,C_i)$ and $\delta_i=I(T_i\leq C_i)$. The random censoring scheme has as special cases the type I and II censoring mechanism. The likelihood function for $\boldsymbol\theta$ is given by
\begin{equation*} 
\begin{aligned}
L(\boldsymbol{\theta},\lambda;\boldsymbol{\mathcal{D}})&=\prod_{i=1}^n g(t_i|\boldsymbol{\theta},\lambda)^{\delta_i}\bar{G}(t_i|\boldsymbol{\theta},\lambda)^{1-\delta_i} \\ &= \left(\frac{\lambda}{1-e^{-\lambda}} \right)^n \prod_{i=1}^n f(t_i|\boldsymbol{\theta})^{\delta_i}e^{-\lambda\delta_i\bar{F}(t_i|\boldsymbol{\theta})}\left(\frac{1- e^{-\lambda \bar{F}(t_i|\boldsymbol{\theta})}}{\lambda}\right)^{1-\delta_i}.
\end{aligned}
\end{equation*}

The log-likelihood function is given by
\begin{equation*} 
\begin{aligned}
l(\boldsymbol{\theta},\lambda;\boldsymbol{\mathcal{D}})=& \,n\log\left(\frac{\lambda}{1-e^{-\lambda}} \right)+ \sum_{i=1}^n\delta_i\log f(t_i|\boldsymbol{\theta}) + \sum_{i=1}^n(1-\delta_i)\log\left(\frac{1- e^{-\lambda \bar{F}(t_i|\boldsymbol{\theta})}}{\lambda}\right) \\& -\lambda\sum_{i=1}^{n}\delta_i\bar{F}(t_i|\boldsymbol{\theta}).
\end{aligned}
\end{equation*}
 
Under the assumption that the likelihood function is differentiable at $\boldsymbol{\theta}$ and $\lambda$. From the partial derivatives of the the log-likelihood function, the likelihood equations are
\begin{equation*} 
\frac{\partial l(\boldsymbol{\theta},\lambda;\boldsymbol{\mathcal{D}})}{\partial \lambda}=\frac{n}{\lambda}+\frac{n}{1-e^{\lambda}} + \sum_{i=1}^n\frac{(1-\delta_i)\left(\lambda\bar{F}(t_i|\boldsymbol{\theta})-e^{\lambda \bar{F}(t_i|\boldsymbol{\theta})}+1\right)}{\lambda\left(e^{\lambda \bar{F}(t_i|\boldsymbol{\theta})}-1\right)} -\sum_{i=1}^{n}\delta_i\bar{F}(t_i|\boldsymbol{\theta}),
\end{equation*}
\begin{equation*} 
\frac{\partial l(\boldsymbol{\theta},\lambda;\boldsymbol{\mathcal{D}})}{\partial \boldsymbol{\theta}}= \sum_{i=1}^n\delta_i  \frac{\partial \log f(t_i|\boldsymbol{\theta})}{\partial \boldsymbol{\theta}} + \sum_{i=1}^n\frac{\lambda(1-\delta_i)}{e^{\lambda \bar{F}(t_i|\boldsymbol{\theta})}-1}\frac{\partial \bar{F}(t_i|\boldsymbol{\theta})}{\partial \boldsymbol{\theta}} -\lambda\sum_{i=1}^{n}\delta_i \frac{\partial \bar{F}(t_i|\boldsymbol{\theta})}{\partial \boldsymbol{\theta}}.
\end{equation*}
Setting the partial derivatives equal to zero, the solutions provide the ML estimates. In many cases, numerical methods such as Newton-Rapshon are required to find the solution of these nonlinear systems.

Under mild conditions the ML estimators of $\boldsymbol{\theta},\lambda$ have an asymptotically Normal joint distribution given by
\begin{equation} (\boldsymbol{\hat{\theta}},\lambda) \sim N[(\boldsymbol{\theta},\lambda),H^{-1}(\boldsymbol{\theta},\lambda)], \mbox{ as } n \to \infty , \end{equation}
where $H(\boldsymbol{\theta},\lambda)$ is the observed information matrix where the elements are given by
\begin{equation*}
\begin{aligned} 
H_{\lambda,\lambda}(\boldsymbol{\theta},\lambda)=&\frac{n}{\lambda^2}-\frac{n e^\lambda}{\left(e^{\lambda}-1\right)^2} + \sum_{i=1}^n(1-\delta_i)\frac{\bar{F}(t_i|\boldsymbol{\theta})e^{\lambda \bar{F}(t_i|\boldsymbol{\theta})}}{\lambda^2\left(e^{\lambda \bar{F}(t_i|\boldsymbol{\theta})}-1\right)^2}- \sum_{i=1}^n\frac{(1-\delta_i) }{t_i^2}\\ & +\sum_{i=1}^{n}\delta_i\bar{F}(t_i|\boldsymbol{\theta}),
\end{aligned}
\end{equation*}
\begin{equation*} 
H_{\lambda,\boldsymbol{\theta}}(\boldsymbol{\theta},\lambda)=- \sum_{i=1}^{n}\delta_i \frac{\partial \bar{F}(t_i|\boldsymbol{\theta})}{\partial \boldsymbol{\theta}}+\sum_{i=1}^n\frac{(1-\delta_i)\left(\lambda \bar{F}(t_i|\boldsymbol{\theta})e^{\lambda \bar{F}(t_i|\boldsymbol{\theta})}-e^{\lambda \bar{F}(t_i|\boldsymbol{\theta})}+1\right)}{\left(e^{\lambda \bar{F}(t_i|\boldsymbol{\theta})}-1\right)^2}\frac{\partial \bar{F}(t_i|\boldsymbol{\theta})}{\partial \boldsymbol{\theta}} ,
\end{equation*}
\begin{equation*} 
\begin{aligned}
H_{\boldsymbol{\theta},\boldsymbol{\theta}}(\boldsymbol{\theta},\lambda)=& \sum_{i=1}^n\frac{\lambda(1-\delta_i)}{\left(e^{\lambda \bar{F}(t_i|\boldsymbol{\theta})}-1\right)^2}\left(\lambda e^{\lambda \bar{F}(t_i|\boldsymbol{\theta})}\left(\frac{\partial \bar{F}(t_i|\boldsymbol{\theta})}{\partial \boldsymbol{\theta}}\right)^2 -\left(e^{\lambda \bar{F}(t_i|\boldsymbol{\theta})}-1 \right)\frac{\partial^2 \bar{F}(t_i|\boldsymbol{\theta})}{\partial^2 \boldsymbol{\theta}} \right) \\ & -\sum_{i=1}^n\delta_i  \frac{\partial^2 \log f(t_i|\boldsymbol{\theta})}{\partial^2 \boldsymbol{\theta}}+\lambda\sum_{i=1}^{n}\delta_i \frac{\partial^2 \bar{F}(t_i|\boldsymbol{\theta})}{\partial^2 \boldsymbol{\theta}} .
\end{aligned}
\end{equation*}

\section{Application}

In this section, we applied our proposed methodology for some common distributions.

\subsection{Exponential distribution}\label{subeep}

Let $Z_1,\ldots,Z_N$ be a non-negative random sample with an exponential distribution where its cdf is given by $F(z;\beta)=1- e^{-\beta z}$, $\beta>0$. Then, using (\ref{cumpe}) it follows that 
\begin{equation}\label{fdppexp} 
g(t;\lambda,\beta)= \frac{\lambda \beta e^{-\beta t-\lambda e^{-\beta t}}}{1 - e^{-\lambda}} ,
\end{equation}
where $\lambda\in \mathbb R^*$ is the shape parameter. Although the p.d.f. has the same form as presented by Cancho et al. \cite{cancho}, by extending the shape parameter into $\mathbb R^*$ we unify the PE distribution with the EP distribution \cite{kus} without adding an additional parameter. More importantly, the shape parameter of this extended exponential Poisson (EEP) distribution has a biological interpretation in terms of CCR problems, i.e., if $\lambda<0$ $(\lambda>0)$ the activation mechanism is the minimum (maximum). Adamidis et al. \cite{adamidis2005extension} discussed a similar interpretation for the shape parameter of the extended exponential geometric distribution which is a unification of the exponential geometric distribution (minimum) \cite{adamidis2} and the complementary exponential geometric distribution (maximum) \cite{louzada2}.

The hazard function of EPE distribution is $h(t)=\lambda \beta e^{-\beta t -\lambda e^{-\beta t}}\left(1-e^{-\lambda e^{-\beta t}}\right)^{-1}$.
Figure \ref{fhazard} gives examples of different shapes for the hazard function.
\begin{figure}[!htb]
	\centering
	\includegraphics[scale=0.5]{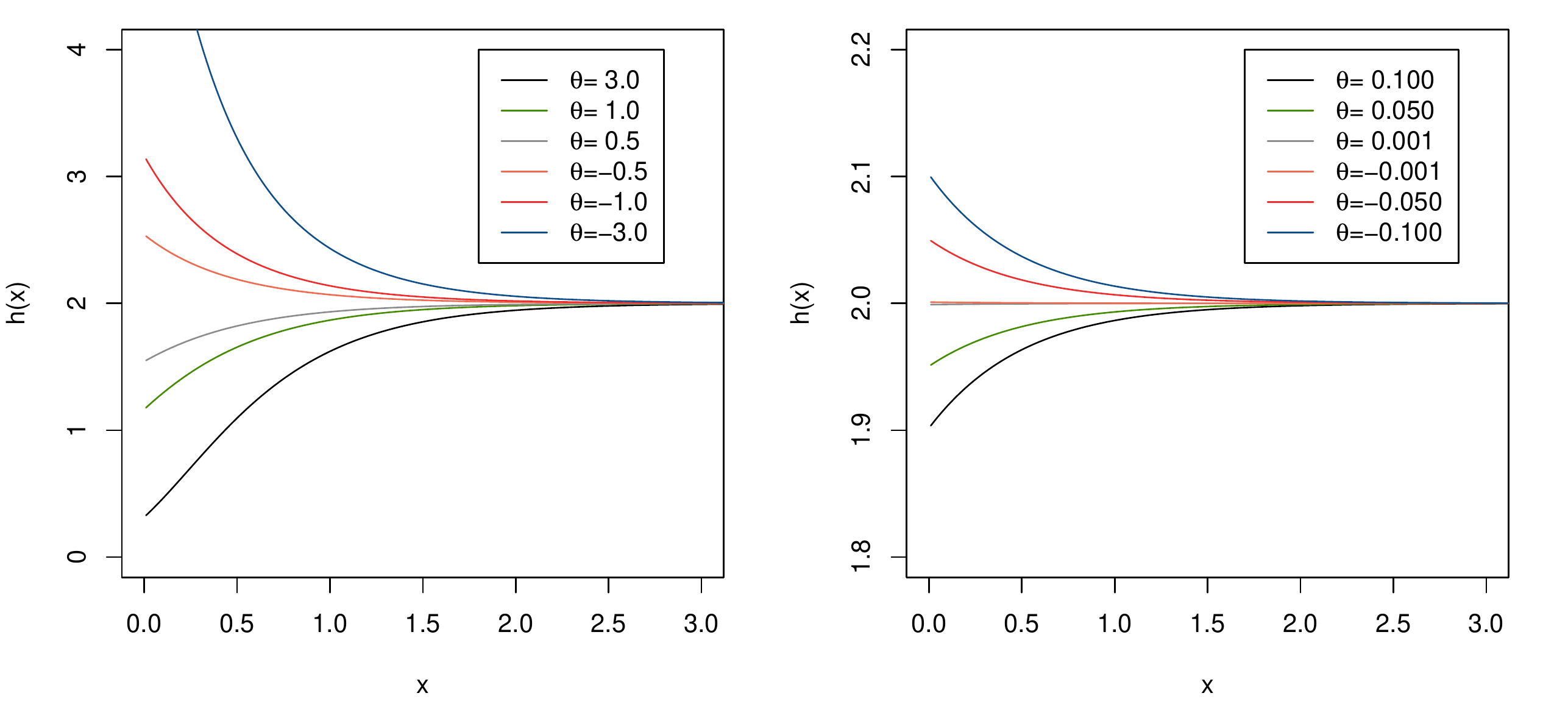}
	\caption{Hazard function shapes for EPE distribution considering different values of $\lambda$ and $\beta=2$.}\label{fhazard}
\end{figure}

Since the exponential distribution has quantile function in closed-form, then using the Proposition \ref{propclo1} the quantile function of the EEP distribution is given by $F^{-1}\left(\f{q}(p,\lambda),\beta\right)=-\beta^{-1}\log\left(1-\f{q}(p,\lambda)\right)$, where $\f{q}(p,\lambda)$ is given in (\ref{qplambda}). Additionally the hazard function of the exponential distribution is given by $h_F(0,\beta)=\beta$ and $0<h_F(0,\beta)<\infty$. Then,  we have $g(0;\lambda,\beta)= \beta\lambda\left(1 - e^{-\lambda}\right)^{-1}\geq 0$, i.e., the EPE distribution also allow the occurrence of zero value.


The maximum likelihood estimators were discussed earlier by  Kus \cite{kus} and Cancho et al. \cite{cancho}. Let $t_1, \ldots, t_n$ be a random sample of size $n$ from EPE distribution, the likelihood function is given by 
\begin{equation}\label{likehoodf1}
L(\lambda, \beta\,;
\boldsymbol{t})=\Bigl(\frac{\lambda \beta}{1-e^{-\lambda}}\Bigl)^n \exp\left\{-\beta \sum_{i=1}^{n}t_i - \lambda \sum_{i=1}^{n}e^{-\beta t_i}\right\} . 
\end{equation}

Cancho et al. \cite{cancho} presented the following log-likelihood function 
\begin{equation*}
\ell(\lambda,\beta; \boldsymbol{t})= n \log(\lambda \beta)-\beta \sum_{i=1}^{n}t_i - \lambda \sum_{i=1}^{n}e^{-\beta t_i}-n \log(1-e^{-\lambda}).
\end{equation*}

Some careful must be taken with the EEG distribution, for instance, $\lambda$ can take negative values then $n\log(\lambda)$ may not be computed. This problem is easily overcome by considering the fact that
\begin{equation*}
\left(\frac{\lambda}{1-e^{-\lambda}} \right)^n>0 \ \ \Leftrightarrow\ \ n\log\left(\frac{\lambda}{1-e^{-\lambda}}\right)\in\mathbb{R}\,,  \ \ \forall \  \lambda\in\mathbb{R}^*.
\end{equation*}

Therefore the log-likelihood function of (\ref{likehoodf1}) is given as
\begin{equation}\label{llikehoodf}
\ell(\lambda,\beta; \boldsymbol{t})= n \log\left(\frac{\lambda}{1-e^{-\lambda}}\right) +n\log(\beta)-\beta \sum_{i=1}^{n}t_i - \lambda \sum_{i=1}^{n}e^{-\beta t_i}.
\end{equation}

\subsection{Weibull distribution}

Now, let $Z_1,\ldots,Z_N$ be a non-negative random sample with cdf given by $F(z;\alpha,\beta)=1- e^{-\beta t^\alpha}$ where $\beta>0$ and $\alpha>0$ . From (\ref{cumpe}) we have 
\begin{equation}\label{fdpweib} 
g(t;\lambda,\alpha,\beta)= \frac{\alpha\lambda\beta  t^{\alpha-1} e^{-\beta t^\alpha-\lambda e^{-\beta t^\alpha}}}{1 - e^{-\lambda}},
\end{equation}
where $\lambda\in\mathbb R^*$. Hemmati et al. \cite{hemmati2011new} discussed a particular case of (\ref{fdpweib}) when $\lambda<0$ (minimum) and named as Weibull-Poisson (WP) distribution. It is worth mentioning that Lu and Shi \cite{lu2012new} independently developed the same distribution and named as WP distribution. Since, our new model also includes the maximum activation mechanism we could be named as extended Weibull-Poisson (EWP) distribution. The hazard function of EWP distribution is $h(t;\lambda,\alpha,\beta)=\lambda \beta t^{\alpha-1}e^{-\beta t^\alpha -\lambda e^{-\beta t^\alpha}}\left(1-e^{-\lambda e^{-\beta t^\alpha}}\right)^{-1}$.
Figure \ref{fhazard-weibull} gives examples of different shapes for the hazard function.
\begin{figure}[!htb]
	\centering
	\includegraphics[scale=0.5]{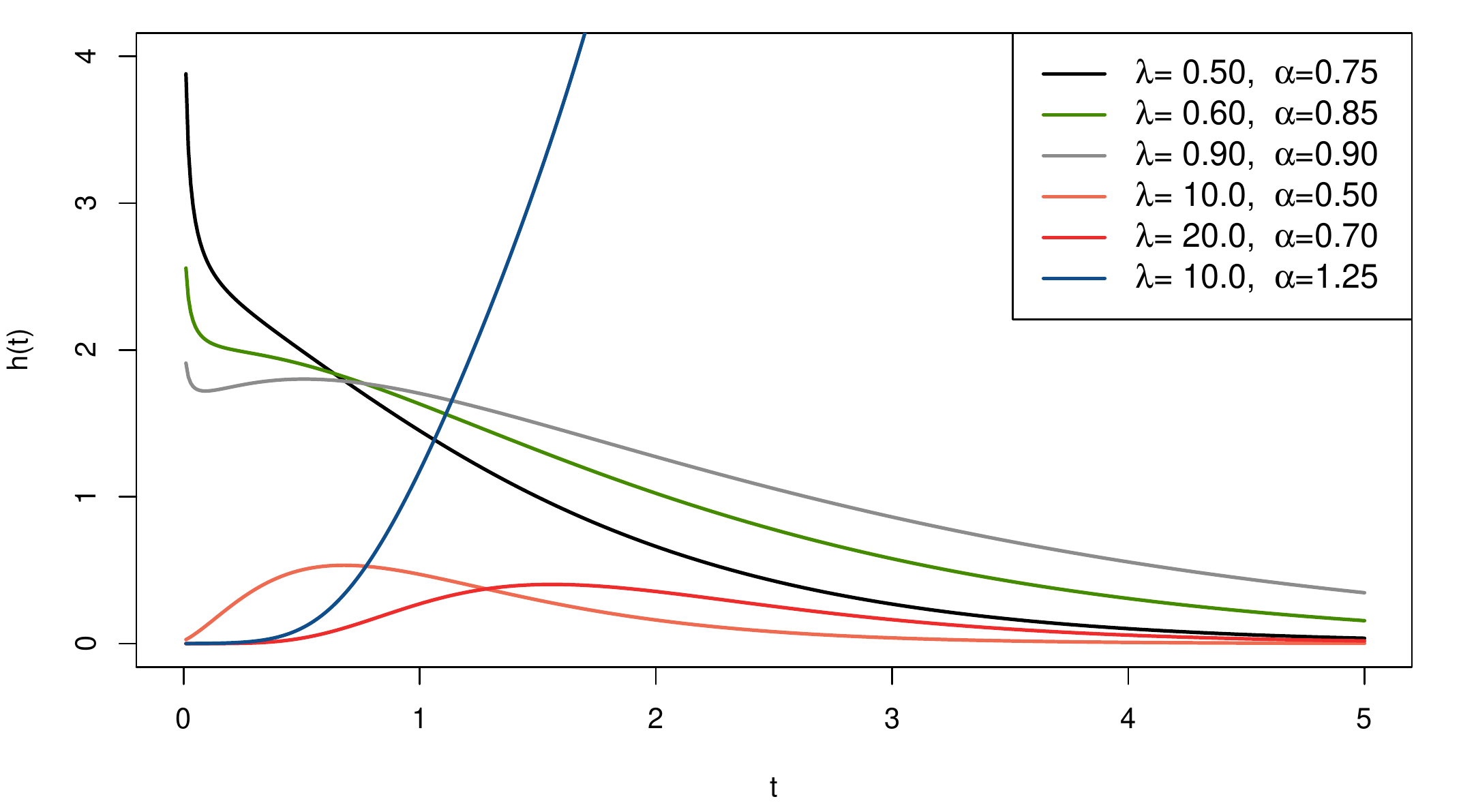}
	\caption{Hazard function shapes for EWP distribution for $\beta=2$  and considering different values of $\lambda$ and $\alpha$.}\label{fhazard-weibull}
\end{figure}

Although Lu and Shi \cite{lu2012new} had shown that WP distribution has increasing and decreasing hazard rate, the extended version also has decreasing-increasing-decreasing and unimodal hazard shape without adding an extra parameter. The Weibull distribution has quantile function in closed-form then using the Proposition \ref{propclo1} the quantile function of the EWP distribution is given by $F^{-1}\left(\f{q}(p,\lambda),\beta,\alpha\right)=\left(-\log\left(1-\f{q}(p,\lambda)\right)/\beta\right)^{1/\alpha}$. Additionally, the hazard function of Weibull distribution is $h_F(t,\alpha,\beta)=\alpha\beta t^{\alpha-1}$ and $0<h_F(0,\alpha,\beta)<\infty$ if and only if $\alpha=1$. Then, the EWP distribution only allow the occurrence of zero value when its reduces to the EEP distribution.

\subsection{Exponentiated exponential-Poisson distribution}

A generalization of exponential-Poisson distribution was proposed by Barreto-Souza and Cribari-Neto \cite{barreto2009generalization} known as generalized exponential-Poisson (GEP) distribution. A random variable T with GEP has the p.d.f. given by
\begin{equation}\label{cumgpe}
f(t;\phi ,\beta,\alpha)=\frac{\alpha\phi\beta}{\left(1-e^{-\phi}\right)^\alpha}\left(1-e^{-\phi+\phi e^{-\beta t}}\right)^{\alpha-1}e^{-\phi-\beta t+\phi e^{-\beta t}},
\end{equation}
where $t, \phi,\beta,\alpha >0$. Note that including a power parameter in (\ref{cumpe}), we have
\begin{equation}\label{cumpepower}
F(t;\boldsymbol{\theta},\lambda,\alpha)=\left(\frac{e^{-\lambda(1- F(t;\boldsymbol{\theta}))}-e^{-\lambda
}}{1-e^{-\lambda}}\right)^\alpha,
\end{equation}
where $\alpha>0$. Following the same procedure described in Section \ref{subeep}, a generalized extended exponential-Poisson (GE2P) has p.d.f. given by
\begin{equation}\label{cumgpe3}
g(t;\beta,\lambda,\alpha)=\frac{\alpha\lambda\beta}{1-e^{-\lambda}}\left(\frac{e^{-\lambda e^{-\beta t}}-e^{-\lambda}}{1-e^{-\lambda}}\right)^{\alpha-1}e^{-\beta t-\lambda e^{-\beta t}},
\end{equation}
where $t,\beta,\alpha >0$ and $\lambda\in \mathbb R^*$. The p.d.f. (\ref{cumgpe3}) is the same as (\ref{cumpepower}) when $\lambda<0$. The current form of the p.d.f. (\ref{cumgpe3}) is important since $1/(1-e^{-\lambda})^\alpha$ may not be defined when $\lambda<0$. The survival function of the GE2P distribution is given by
\begin{equation*}
\bar{G}(t;\beta ,\lambda,\alpha)=1-\left(\frac{e^{-\lambda e^{-\beta t}}-e^{-\lambda}}{1-e^{-\lambda}}\right)^\alpha .
\end{equation*}

The hazard function is obtained from $h(t,\lambda ,\beta,\alpha)=g(t,\beta,\lambda ,\alpha)/\bar{G}(t,\beta,\lambda,\alpha)$. Barreto-Souza and Cribari-Neto \cite{barreto2009generalization} proved that the hazard function has decreasing, 
increasing  or unimodal shape (for $\lambda<0$). Although we have the same number of parameters, when $\lambda>0$ the hazard function of the GE2P can have bathtub shape (see Figure \ref{fhazard-bar-crib}).

\begin{figure}[!htb]
	\centering
	\includegraphics[scale=0.5]{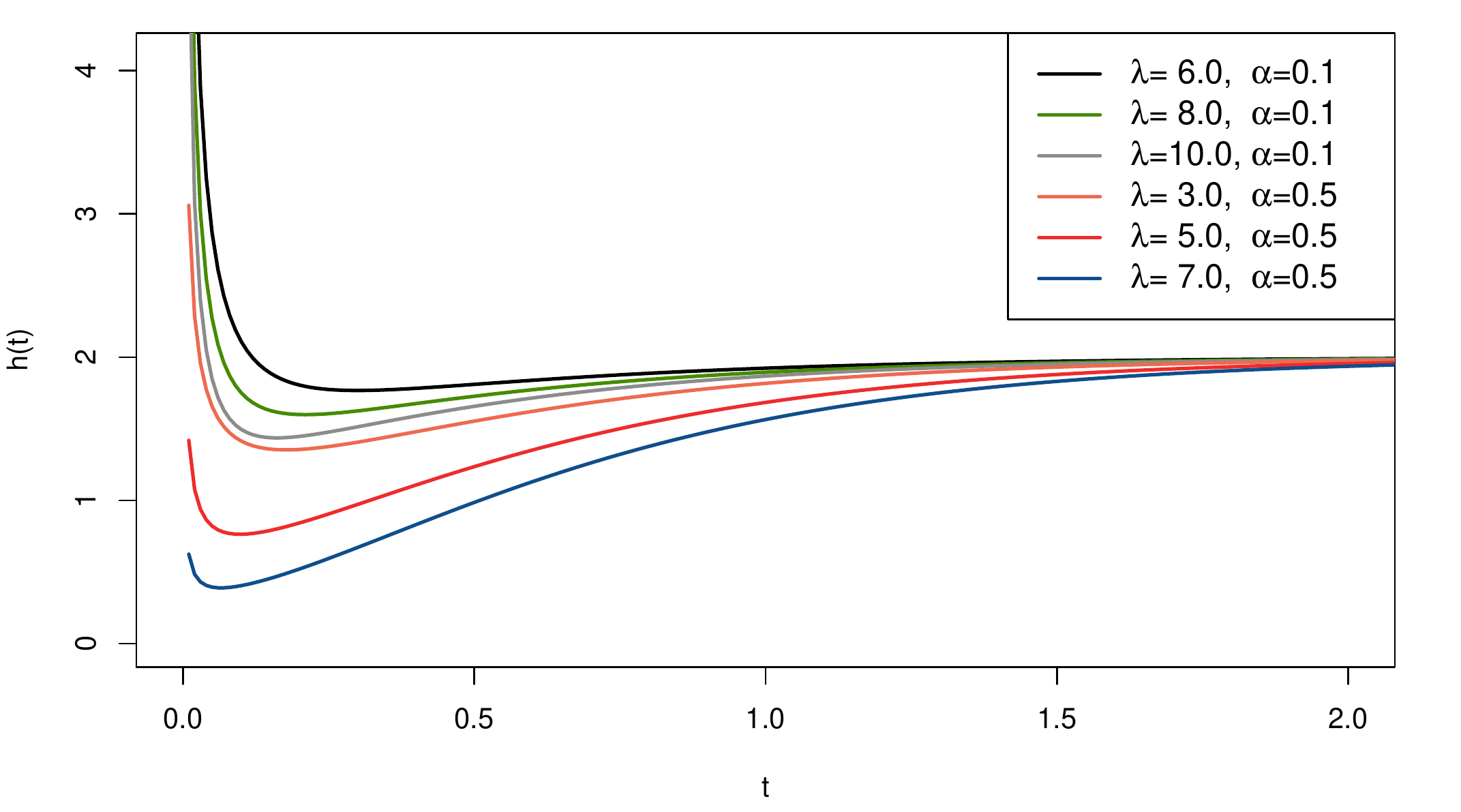}
	\caption{Hazard function shapes for GE2P distribution for $\beta=2$  and considering different values of $\lambda$ and $\alpha$.}\label{fhazard-bar-crib}
\end{figure}
Therefore, this simple extension of the GEP distribution has the hazard function with decreasing, 
increasing, bathtub or unimodal shape. Note that, the EGEP distribution is a double compounded distribution in which we applied firstly our approach in the exponential distribution and secondly the Lehmann \cite{lehmann1953power} approach.  Therefore, since EPE distribution allow the occurrence of zero value the EGEP distribution has the same property. The quantile function of the EGEP distribution is given by $F^{-1}\left(\f{q}(p,\lambda),\beta,\alpha\right)=-\beta^{-1}\log\left(1-\f{q}(p,\lambda)^{1/\alpha}\right)$ where $\f{q}(p,\lambda)$ is given in (\ref{qplambda}).

\subsection{Generalized extreme value distribution}

The Generalized extreme value (GEV) distribution plays an important role in extreme value theory for modeling rare events. The GEV distribution \cite{jenkinson} has as special cases the Gumbel, Fr\'echet and Weibull distribution, its cdf is given by 
\begin{equation}\label{gev}
F(t|\sigma,\mu,\xi)=
\begin{cases}
\exp\left\{ -\left[1+\xi(t-\mu)/\sigma\right]_+^{-1/\xi}\right\}, & \xi\neq 0 \\
\exp\left\{ -\exp\left[-(t-\mu)/\sigma\right]\right\}, & \xi= 0,
\end{cases}
\end{equation}
where $\sigma>0$, $\mu, \xi \in \mathbb{R}$ and $t_+=\max(t,0)$. The  extended  generalized extreme value Poisson (EGEVP) distribution has the cdf given by
\begin{equation}\label{cumpgev}
G(t;\lambda,\sigma,\mu,\xi)=\frac{e^{\lambda e^{-r(t;\sigma,\mu,\xi)}}-1}{e^{\lambda}-1},
\end{equation}
where
\vspace{-0.6cm}
\begin{equation*}
r(t;\sigma,\mu,\xi)=
\begin{cases}
\left[1+\xi(t-\mu)/\sigma\right]_+^{-1/\xi}, & \xi\neq 0  \\
\exp\left[-(t-\mu)/\sigma\right], & \xi= 0.
\end{cases}
\end{equation*}

The pdf of the EGEVP distribution has the following form
\begin{equation}\label{cdfgev}
g(t;\lambda,\sigma,\mu,\xi)=\frac{\lambda\,r(t;\sigma,\mu,\xi)^{\xi+1}}{\left(1-e^{-\lambda}\right)\sigma}\exp\left\{-\lambda(1-  e^{-r(t;\sigma,\mu,\xi)}))-r(t;\sigma,\mu,\xi)\right\}.
\end{equation}

The EGEVP distribution (\ref{cdfgev}) has as special cases the EEP distribution, EWP distribution, extended Gumbel-Poisson distribution, extended Fr\'echet-Poisson distribution, to list a few. 
The quantile function has closed-form and is given by
\begin{equation}\label{gev}
F^{-1}\left(\f{q}(p,\lambda),\sigma,\mu,\xi\right)=
\begin{cases}
\mu +\frac{\sigma}{\xi}\left(\left[-\log(\f{q}(p,\lambda))\right]^{-\xi}-1 \right), & \xi\neq 0 \\
\mu -\sigma\log\left(-\log(\f{q}(p,\lambda)) \right), & \xi= 0 .
\end{cases}
\end{equation}

The hazard function is given by
\begin{equation}
h(t;\boldsymbol{\theta},\lambda)=\frac{\lambda\, r(t;\sigma,\mu,\xi)^{\xi+1}e^{-r(t;\sigma,\mu,\xi)}}{\sigma\left(e^{\lambda(1-  \exp\{-r(t;\sigma,\mu,\xi)\})}-1\right)}.
\end{equation}
Figure \ref{fhazard-gev} gives examples of different shapes for the hazard function in which allows us to fit data with increasing, decreasing and unimodal hazard rate.
\begin{figure}[!htb]
	\centering
	\includegraphics[scale=0.5]{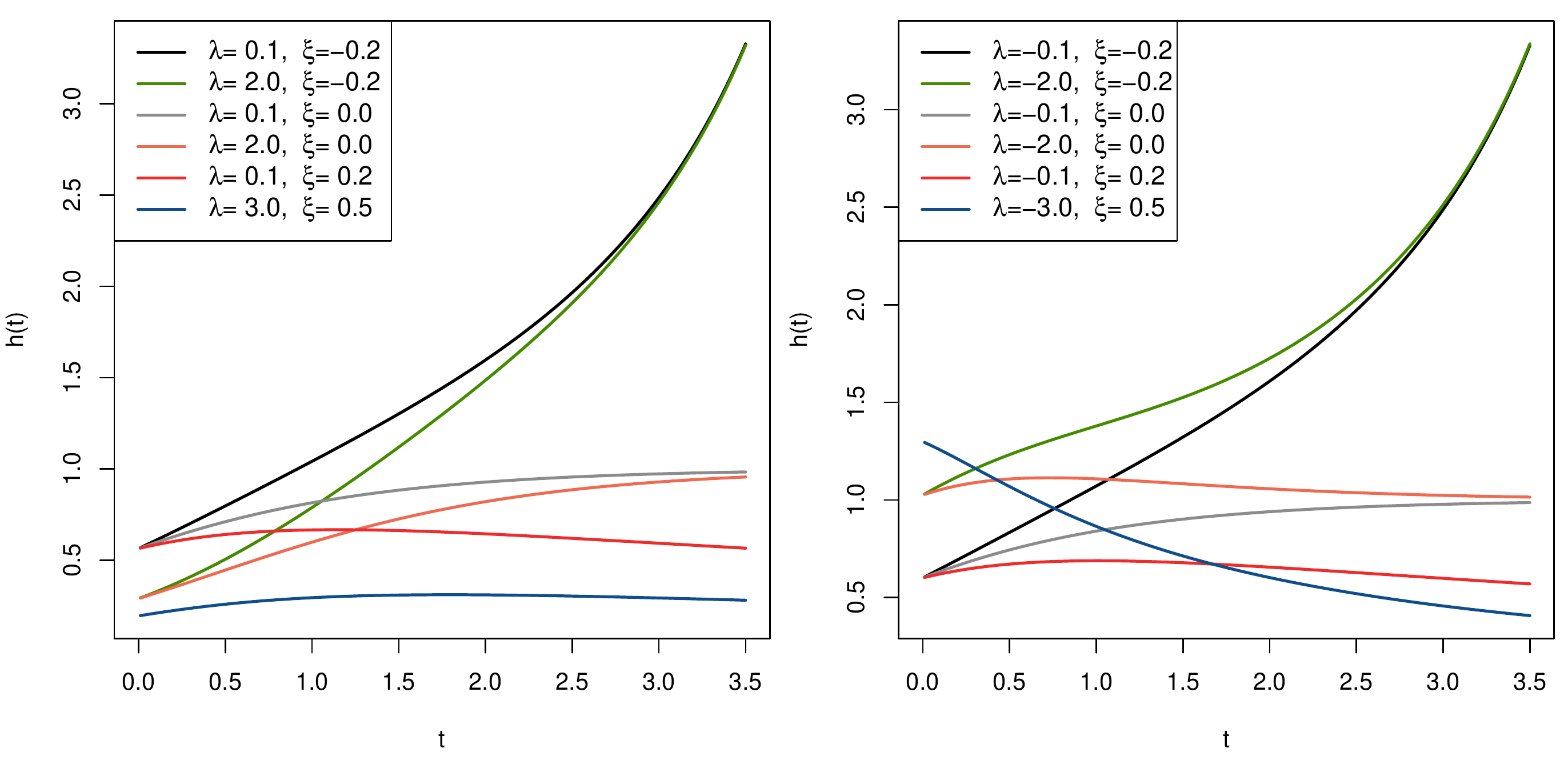}
	\caption{Hazard function shapes for EGEVP distribution for $\mu=0,\sigma=1$  and considering different values of $\lambda$ and $\xi$.}\label{fhazard-gev}
\end{figure}


\subsection{Other compound models	based on the Poisson distribution}

Many distributions have already been proposed using the minimum or maximum activation mechanism. Tahir and Cordeiro \cite{tahir2016compounding} presented an interesting discussion about compounding different distributions. They reviewed some already introduced distributions based on the zero truncated Poisson distribution such as: modified EP, exponentiated EP, Beta-Weibull Poisson, complementary modified Weibull-Poisson, complementary exponentiated Weibull-Poisson, Log-logistic generalized Weibull Poisson, Lai-modified Weibull-Poisson, Exponentiated Lomax-Poisson, complementary Poisson-Lomax, Lindley-Poisson, complementary extended Lindley-Poisson, Poisson Birnbaum-Saunders, exponentiated-Burr XII Poisson, complementary Burr III Poisson, complementary failure rate Poisson and the complementary exponentiated power Lindley-Poisson distribution (see \cite{tahir2016compounding} and references therein). Our approach can be applied in any of these distributions unifying the minimum/maximum without including an extra parameter.

This approach can be applied in various models used in other areas. For instance, Macera et al. \cite{macera2015exponential} introduced a new model for recurrent event data characterized by a baseline rate function based on the exponential-Poisson distribution. Following Zhao and Zhou \cite{zhao2012modeling}, the authors presented a rate model which
is derived from a nonhomogeneous Poisson process, with a hazard rate function $h(w;t)=h_0(w+t)$. The rate function of recurrence process $N(w+t_{j-1})$ is given by
\begin{equation}\label{raterede2p}
h(w;t_{j-1},\phi, \beta)=\frac{\phi\beta e^{-\beta(w+t_{j-1})}}{1-\exp\left(-\phi e^{-\beta(w+t_{j-1})}\right)},
\end{equation}
where $w, \phi, \beta>0$. The rate function (\ref{raterede2p}) has decreasing behavior. Further, Louzada et al. \cite{louzada2015poisson} developed a similar study based on the Poisson-exponential distribution. Both models can be unified in one where the p.d.f of the new distribution for recurrent event data where baseline rate function is EPE distribution is given by
\begin{equation}\label{daterede2p}
f(w;t_{j-1},\phi, \beta)=\frac{\phi\beta\exp\left(-\beta(w+t_{j-1})\right)+\phi e^{-\beta(w+t_{j-1})}}{\exp\left(\phi e^{-\beta t_{j-1}} \right)-1},
\end{equation}
where $w, \beta>0$ and $\phi\in \mathbb R^*$.  The hazard function for $\phi>0$ ($\phi<0$) has decreasing (increasing) shape.

\begin{table}[!h]
\centering
\caption{The Bias(MSE) for for the estimates of $\boldsymbol{\theta}$ and $\lambda$ considering different sample sizes.}\vspace{3mm}
{\footnotesize
\begin{tabular}{c|r|r|r|r|r|r|r|r}
\hline
\multicolumn{1}{c|}{Model} & \multicolumn{1}{c|}{C} & \multicolumn{1}{c|}{n} &  \multicolumn{1}{c|}{$\lambda_1$} & \multicolumn{1}{c|}{$\beta_1$} & \multicolumn{1}{c|}{$\alpha_1$} & \multicolumn{1}{c|}{CP($\lambda_1$)} & \multicolumn{1}{c|}{CP($\beta_1$)} & \multicolumn{1}{c}{CP($\alpha_1$)} \\
\hline
\multirow{10}{*}{EEP} &          & 50  & 0.2549(0.9006) & 0.1708(0.4473) & ------------------ & 97.90$\%$ & 96.00$\%$ & --------- \\
                              &  & 100 & 0.1064(0.4753) & 0.0711(0.2316) & ------------------ & 96.50$\%$ & 94.80$\%$ & --------- \\
                        & $30\%$ & 200 & 0.0163(0.2392) & 0.0163(0.1123) & ------------------ & 95.60$\%$ & 95.10$\%$ & --------- \\
                               & & 300 & 0.0243(0.1647) & 0.0177(0.0754) & ------------------ & 94.90$\%$ & 94.70$\%$ & --------- \\ 
                               & & 400 & 0.0168(0.1221) & 0.0143(0.0562) & ------------------ & 94.70$\%$ & 94.70$\%$ & --------- \\  \cline{2-9}
                                & & 50 & 0.2761(1.0738) & 0.1864(0.6596) & ------------------ & 97.50$\%$ & 95.60$\%$ & --------- \\                              &  & 100 & 0.1630(0.5745) & 0.1154(0.3599) & ------------------ & 96.90$\%$ & 95.90$\%$ & --------- \\
                        & $50\%$ & 200 & 0.0192(0.3256) & 0.0174(0.1979) & ------------------ & 96.80$\%$ & 95.50$\%$ & --------- \\ 
                               & & 300 & 0.0056(0.2160) & 0.0055(0.1279) & ------------------ & 96.20$\%$ & 95.40$\%$ & --------- \\ 
                               & & 400 & 0.0056(0.1684) & 0.0063(0.0985) & ------------------ & 95.20$\%$ & 94.80$\%$ & --------- \\ 
\hline
\multirow{10}{*}{EWP} &           & 50 &  0.0265(2.2973) & -0.0057(0.2259) & 0.0306(0.0227) & 99.10$\%$ & 97.90$\%$ & 97.10$\%$ \\ 
                              &  & 100 & -0.3926(1.5546) & -0.1251(0.1515) & 0.0478(0.0185) & 97.90$\%$ & 97.60$\%$ & 95.90$\%$ \\ 
                        & $30\%$ & 200 & -0.1407(0.9269) & -0.0464(0.0874) & 0.0174(0.0089) & 98.70$\%$ & 98.60$\%$ & 96.10$\%$ \\ 
                               & & 300 & -0.1813(0.7587) & -0.0605(0.0760) & 0.0180(0.0065) & 98.20$\%$ & 97.90$\%$ & 96.40$\%$ \\ 
                               & & 400 & -0.1584(0.5957) & -0.0536(0.0617) & 0.0165(0.0050) & 98.60$\%$ & 98.10$\%$ & 97.20$\%$ \\   \cline{2-9}
                                & & 50 &  0.8497(3.7558) &  0.1942(0.2694) &-0.0259(0.0212) & 99.80$\%$ & 97.40$\%$ & 96.70$\%$ \\                              &  & 100 & -0.2244(1.5073) & -0.0641(0.1348) & 0.0393(0.0182) & 99.30$\%$ & 99.00$\%$ & 97.30$\%$ \\ 
                        & $50\%$ & 200 & -0.0646(1.0365) & -0.0284(0.0839) & 0.0174(0.0121) & 98.90$\%$ & 98.80$\%$ & 96.10$\%$ \\ 
                               & & 300 & -0.1240(0.9821) & -0.0443(0.0779) & 0.0189(0.0101) & 98.30$\%$ & 97.90$\%$ & 94.90$\%$ \\ 
                               & & 400 & -0.2153(0.8449) & -0.0696(0.0775) & 0.0244(0.0088) & 98.00$\%$ & 98.10$\%$ & 95.00$\%$ \\ 
\hline
\multirow{10}{*}{G2EP} &         &  50 & -0.3867(8.8114) & -0.1703(0.2602) &  0.3855(1.9715) & 96.00$\%$ & 95.40$\%$ & 95.10$\%$ \\
                              &  & 100 & -0.0695(4.8959) & -0.1201(0.1265) &  0.0865(0.6004) & 96.10$\%$ & 96.50$\%$ & 93.50$\%$ \\
                        & $30\%$ & 200 &  0.0706(2.3952) & -0.0865(0.0577) & -0.0558(0.2426) & 95.20$\%$ & 96.00$\%$ & 92.90$\%$ \\ 
                               & & 300 &  0.1075(1.5907) & -0.0766(0.0381) & -0.1032(0.1592) & 96.30$\%$ & 95.30$\%$ & 92.00$\%$ \\ 
                               & & 400 &  0.1123(1.1818) & -0.0726(0.0292) & -0.1239(0.1231) & 96.50$\%$ & 94.80$\%$ & 91.10$\%$ \\  \cline{2-9} 
                               & &  50 & -0.6999(10.6788) & -0.2683(0.4341) &  0.4395(2.1936) & 97.60$\%$ & 95.40$\%$ & 95.60$\%$ \\                              &  & 100 & -0.3131(6.5059)  & -0.1973(0.2485) &  0.1273(0.6688) & 98.10$\%$ & 96.10$\%$ & 93.70$\%$ \\
                        & $50\%$ & 200 & -0.0891(3.2484)  & -0.1363(0.1185) & -0.0217(0.2730) & 98.00$\%$ & 97.20$\%$ & 92.90$\%$ \\ 
                               & & 300 & -0.0296(2.0539)  & -0.1150(0.0748) & -0.0707(0.1774) & 96.90$\%$ & 96.90$\%$ & 92.70$\%$ \\ 
                               & & 400 & -0.0032(1.5007)  & -0.1043(0.0550) & -0.0956(0.1347) & 96.50$\%$ & 96.10$\%$ & 92.10$\%$ \\ 
\hline

\end{tabular}
}
\label{tableres1}
\end{table}

\section{Numerical evaluation}

In this section a simulation study is presented to in order to check the efficiency of the ML estimates under random censoring by computing the bias and the mean square errors (MSE), given by
\begin{equation*}
\f{Bias}_{i}\frac{1}{N}\sum_{j=1}^{N}\left(\hat\Theta_{i,j}-\Theta_i\right) \ \ \mbox{ and } \ \ \f{MSE}_{i}=\sum_{j=1}^{N}\frac{(\hat\Theta_{i,j}-\Theta_i)^2}{N}, \quad i=1,\ldots, k,
\end{equation*} 
where $\boldsymbol{\Theta}=(\boldsymbol{\theta},\lambda)$ are the parameters related to $G(t\,;\boldsymbol{\Theta})$ and $N=200,000$ is the number of estimates obtained through the ML estimators. Under this approach, the best estimators should provide both Bias and MSE closer to zero. In addition, the $95\%$ coverage probability (CP$_{95\%}$) of the confidence intervals are also evaluated in which for a large $N$ under $95\%$ confidence level, the frequencies of intervals that covered the true values of $\boldsymbol{\Theta}$ should be closer to $95\%$. 

The simulation study was carry out using the software R and the sample sizes were $n=(50,100,200,400)$. The distributions used in the simulation study are the EP, EW and the G2EP distribution. The chosen values to perform this study were $(\lambda_1,\beta_1)=(2,3)$ and $(\lambda_2,\beta_2)=(-5,0.5)$ for the EP distribution, $(\lambda_1,\beta_1,\alpha_1)=(3,2,0.5)$ and $(\lambda_2,\beta_2,\alpha_2)=(-3,0.3,1.5)$ for the EW distribution and $(\lambda_1,\beta_1,\alpha_1)=(2,2,1.5)$ and $(\lambda_2,\beta_2,\alpha_2)=(-1,0.2,0.8)$ for the G2EP distribution. However, the following results were similar for other choices of $\boldsymbol{\Theta}$. The samples were generated using random sampling with respectively $30\%$ and $50\%$ of censoring. Tables \ref{tableres1} and \ref{tableres2} present the Bias, MSEs and $CP_{95\%}$ for the obtained estimates.

\begin{table}[t]
\centering
\caption{The Bias(MSE) for for the estimates of $\boldsymbol{\theta}$ and $\lambda$ considering different sample sizes.}\vspace{3mm}
{\footnotesize
\begin{tabular}{c|r|r|r|r|r|r|r|r}
\hline
\multicolumn{1}{c|}{Model} & \multicolumn{1}{c|}{C} & \multicolumn{1}{c|}{n} &  \multicolumn{1}{c|}{$\lambda_1$} & \multicolumn{1}{c|}{$\beta_1$} & \multicolumn{1}{c|}{$\alpha_1$} & \multicolumn{1}{c|}{CP($\lambda_1$)} & \multicolumn{1}{c|}{CP($\beta_1$)} & \multicolumn{1}{c}{CP($\alpha_1$)} \\
\hline
\multirow{10}{*}{EEP} &          & 50  &  0.1919(2.5781) & 0.1175(0.1228) & ------------------ & 98.10$\%$ & 98.10$\%$ & --------- \\
                              &  & 100 &  0.0581(2.0995) & 0.0625(0.0520) & ------------------ & 99.20$\%$ & 97.70$\%$ & --------- \\
                        & $30\%$ & 200 & -0.3527(2.4040) & 0.0085(0.0410) & ------------------ & 96.90$\%$ & 94.20$\%$ & --------- \\ 
                               & & 300 & -0.3098(1.2802) &-0.0094(0.0168) & ------------------ & 99.10$\%$ & 94.60$\%$ & --------- \\ 
                               & & 400 & -0.0905(1.4225) & 0.0220(0.0263) & ------------------ & 97.30$\%$ & 97.30$\%$ & --------- \\  \cline{2-9}
                                 & & 50 &  0.2003(2.6696) &  0.1170(0.1074) & ------------------ & 99.90$\%$ & 97.80$\%$ & --------- \\                               &  & 100 &  0.1508(2.6916) &  0.0827(0.0718) & ------------------ & 99.80$\%$ & 99.10$\%$ & --------- \\
                         & $50\%$ & 200 &  0.1447(2.9435) &  0.0908(0.0892) & ------------------ & 98.00$\%$ & 97.00$\%$ & --------- \\ 
                                & & 300 & -0.4525(1.9526) & -0.0168(0.0207) & ------------------ & 96.60$\%$ & 91.90$\%$ & --------- \\ 
                                & & 400 & -0.2489(1.5965) &  0.0053(0.0258) & ------------------ & 96.10$\%$ & 94.30$\%$ & --------- \\ 
\hline
\multirow{10}{*}{EWP} &           & 50 & 1.3965(6.4446) & 0.3083(0.2869) &-0.0231(0.0654) & 92.00$\%$ & 96.20$\%$ & 95.80$\%$ \\
                              &  & 100 & 0.8055(2.5338) & 0.1431(0.0854) & 0.0093(0.0284) & 95.50$\%$ & 98.50$\%$ & 96.60$\%$ \\
                        & $30\%$ & 200 & 0.4309(1.1024) & 0.0653(0.0249) & 0.0112(0.0133) & 97.80$\%$ & 99.40$\%$ & 96.10$\%$ \\ 
                               & & 300 & 0.1442(0.6853) & 0.0252(0.0112) & 0.0094(0.0092) & 97.70$\%$ & 99.60$\%$ & 95.20$\%$ \\ 
                               & & 400 & 0.1388(0.6151) & 0.0236(0.0101) & 0.0025(0.0066) & 97.50$\%$ & 99.40$\%$ & 95.00$\%$ \\  \cline{2-9}
                               & &  50 & 2.8575(11.6493) & 0.6097(0.6121) & -0.0465(0.1151) & 90.10$\%$ & 95.40$\%$ & 96.00$\%$ \\                              &  & 100 & 1.6103(5.9828)  & 0.3198(0.2502) & -0.0131(0.0499) & 92.10$\%$ & 97.30$\%$ & 95.60$\%$ \\
                        & $50\%$ & 200 & 1.0107(3.1389)  & 0.1830(0.1148) &  0.0089(0.0209) & 93.10$\%$ & 97.80$\%$ & 96.30$\%$ \\ 
                               & & 300 & 0.5147(1.5409)  & 0.0867(0.0422) &  0.0092(0.0138) & 96.50$\%$ & 98.40$\%$ & 94.90$\%$ \\ 
                               & & 400 & 0.2043(0.8367)  & 0.0370(0.0175) &  0.0086(0.0102) & 98.00$\%$ & 98.40$\%$ & 94.90$\%$ \\ 
\hline
\multirow{10}{*}{G2EP} &         &  50 & 0.4124(1.2689) & 0.0569(0.0148) &  0.0200(0.0349) & 98.00$\%$ & 96.00$\%$ & 95.90$\%$ \\ 
                              &  & 100 & 0.2857(0.6519) & 0.0327(0.0063) & -0.0039(0.0155) & 97.50$\%$ & 96.40$\%$ & 94.90$\%$ \\ 
                        & $30\%$ & 200 & 0.1981(0.3934) & 0.0229(0.0036) & -0.0058(0.0105) & 96.90$\%$ & 97.00$\%$ & 95.30$\%$ \\ 
                               & & 300 & 0.1786(0.3091) & 0.0190(0.0026) & -0.0108(0.0081) & 96.10$\%$ & 96.90$\%$ & 94.50$\%$ \\  
                               & & 400 & 0.0712(0.2904) & 0.0098(0.0023) & -0.0144(0.0062) & 95.50$\%$ & 95.60$\%$ & 91.10$\%$ \\   \cline{2-9} 
                                & & 50 & 0.9410(2.4972) & 0.1239(0.0400) &  0.0123(0.0410) & 97.60$\%$ & 95.10$\%$ & 95.40$\%$ \\                              &  & 100 & 0.5445(1.2879) & 0.0667(0.0159) & -0.0012(0.0195) & 96.30$\%$ & 96.70$\%$ & 95.30$\%$ \\
                        & $50\%$ & 200 & 0.4949(1.0424) & 0.0552(0.0116) & -0.0158(0.0136) & 94.30$\%$ & 95.40$\%$ & 94.70$\%$ \\ 
                               & & 300 & 0.4626(0.8084) & 0.0503(0.0088) & -0.0162(0.0107) & 93.10$\%$ & 95.40$\%$ & 93.40$\%$ \\ 
                               & & 400 & 0.4259(0.5968) & 0.0408(0.0057) & -0.0181(0.0077) & 91.60$\%$ & 95.30$\%$ & 90.50$\%$ \\ 
\hline

\end{tabular}
}
\label{tableres2}
\end{table}

From these results, we observed that both Bias and MSE  tend to zero as there is an increase of n, i.e., the ML estimators are asymptotically unbiased for the parameters. Moreover, the coverage probability tends to the nominal level as n increase. Other estimation procedures can be considered for these models. For instance, Rodrigues et al. \cite{rodrigues2016poisson} compared ten different estimation methods for the parameters of PE distribution under complete data and concluded that a minimum distance estimator provided better results than the ML estimators, a similar study can be conducted in the presence of censored data and for the other models. Additionally, it is important to point out that a simulation study was not presented for EGEVP, since its ML estimators showed to be non-identifiable, leading to different roots depending on the data set, in this case the conditions for the asymptotic properties were not fulfill and confidence intervals could not be constructed. Further research are need considering other estimation procedures for this particular model.

\newpage

\section{Application}

In this section we considered a data set related to failure time of devices of an airline company. The study of its failure can prevent
customer dissatisfaction and customer attrition which avoid company loss. Table 1 presents the data related to failure time of (in days) of $131$ devices in an aircraft (+ indicates the presence of censorship).

\begin{table}[!h]
	\caption{Data set related to the failure time of $131$ devices in an aircraft.}\vspace{3mm}
\centering 
	{\begin{tabular}{c c c c c c c c c c c c c } 
		\hline 
	 36 & 15 &  4 &  9 & 20 &  2 & 127+ & 97 & 28 & 22 & 329+ & 158 & 43 \\
	 45 & 21 & 24 &  9 & 84+ & 237 & 56+ & 18 &  2 &  1 &  2 &  9 &  4 \\ 
	1 &  2 &  1 & 19 & 20 &  3 &  1 &  2 &  1 &  6 & 10 &  7 & 33 \\ 
	16 &  2 & 17 & 10 &  8 & 30+ & 25 & 13 & 36 &  7 &  2 &  2 & 93 \\
	44 &  3 &  3 & 12 & 11 &  1 & 15 & 16 &  2 & 18 & 10 & 18 & 76 \\
	16 & 92 &  3 & 28 & 53+ & 29 & 46 & 11 & 94 & 95 &  1 & 33 & 40 \\
	22 & 12 & 15 & 46 & 20 & 53+ & 74 & 126 & 27 & 14 & 22 & 79 & 15 \\
		8 & 68 & 81+ & 51 &  7 &  2 & 20 & 24 & 11 & 16 &  3 & 42 &  2 \\ 
		10 & 52 &  5 & 46 &  5 & 37 & 14 & 40 & 95+ & 24 & 10 &  3 & 20 \\
		167+ & 44 &  8 &  1 & 18 & 28 & 17 & 11 & 10 & 16 & 79 & 20 & 55 \\
		115+\\ [0ex] 
		\hline 
	\end{tabular}}\label{tableairplane}
\end{table}

Different discrimination criterion methods based on log-likelihood function evaluated at the ML estimates were also considered. The discrimination criterion methods are respectively: Akaike information criterion (AIC) computed through $\f{AIC}=-2l(%
\boldsymbol{\hat{\Theta}};\boldsymbol{t})+2k$ and the corrected Akaike information
criterion $\f{AICC}=\f{AIC}+\frac{2\,k\,(k+1)}{(n-k-1)}$, where $k$ is the number of
parameters to be fitted and $\boldsymbol{\hat{\Theta}}$ is the estimates of $\boldsymbol{\Theta}$. The best model is the one which provides the minimum values
of those criteria. Table \ref{discairplane} presents the results of AIC, AICc criteria, for different probability distributions.

\begin{table}[ht]
	\caption{Results of AIC, AICc criteria for different probability distributions considering the data set related to the failure time of $131$ of devices in an aircraft.}\vspace{3mm}
\centering 
	{\begin{tabular}{c|c|c|c|c}
			\hline
			Test   &  EP & EW & GE2P & EGEV  \\ \hline
			AIC    & 1084.38 & \textbf{1084.04} & 1085.73 & 1085.75 \\ 
			AICc   & 1084.48 & \textbf{1084.22} & 1086.05 & 1085.94 \\ \hline
		\end{tabular}}\label{discairplane}
\end{table}

Figure \ref{grafico-obscajust1} presents the survival function adjusted by different distributions and the Kaplan-Meier estimator.

\begin{figure}[!htb]
	\centering
	\includegraphics[scale=0.5]{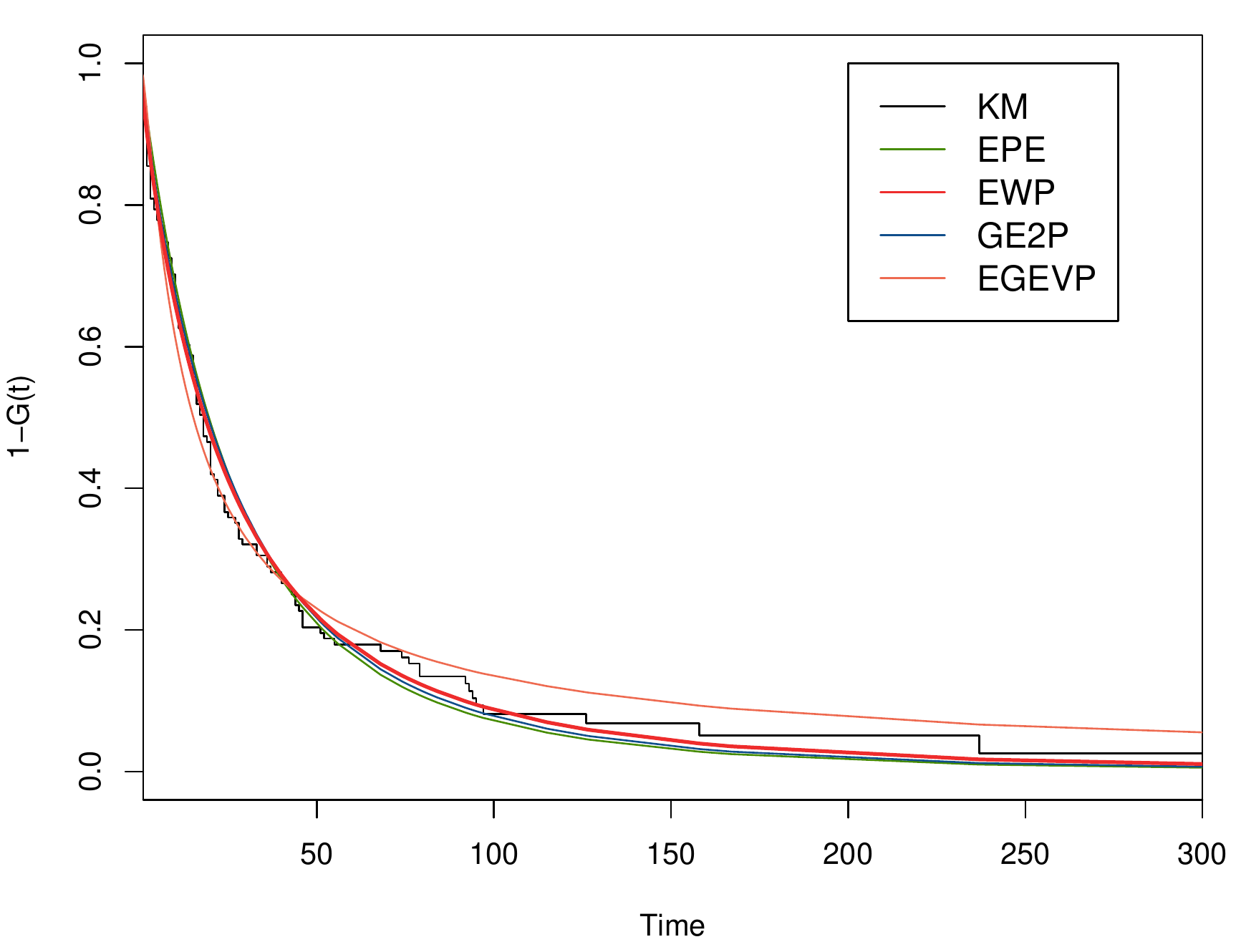}
	\caption{Survival function adjusted by different distributions and the Kaplan-Meier estimator considering data set related to the failure time of $131$ devices.}\label{grafico-obscajust1}
\end{figure}

Comparing the empirical survival function with the adjusted models we observed a goodness of  the fit EW distribution.Additionaly, from the results obtained by the AIC, AICc the EW returned the minimum value, i.e., among the proposed models the EW distribution fits better the data related to the failure time of aircraft devices. Table \ref{resairplanetab} displays the ML estimates, standard-error (SE) and  $95\%$ confidence intervals for $\lambda$, $\beta$ and $\alpha$.

\begin{table}[ht]
	\caption{ML estimates (MLE), SE (standard errors) and  $95\%$ confidence intervals for $\lambda$, $\beta$ and $\alpha$}\vspace{3mm}
	\centering 
{\small
	{\begin{tabular}{ c | c |  c| c }
			\hline
			$\boldsymbol{\Theta}$  & MLE & SE & $CI_{95\%}(\boldsymbol{\Theta})$ \\ \hline
			\ \ $\lambda$   \ \  &  -3.68674 &   3.09309 &  (-7.13376;  -0.23971) \\ \hline
				\ \ $\beta$ \ \   & 0.01463 &  0.00004 &  (0.00189; 0.02737)  \\ \hline
			\ \ $\alpha$   \ \  &  0.89760 &  0.00466 &  (0.76379; 1.03141) \\ \hline
		\end{tabular}}\label{resairplanetab} }
\end{table}

Under this approach the data set allow us to discovery if the activation mechanism comes from the minimum or maximum. Since the ML estimates of $\lambda$ returned negative value we concluded the activation mechanism comes from the minimum of Weibull distributions, i.e., if $T_i, \, i=1,\ldots,131$ is our data set than $T_i$ represents the lifetime of the minimum of $X_{i,j}, j=1,\ldots,N$  where $X_{i,j}$ follows a Weibull distribution and N is random and not observable. 

\section{Concluding remarks}

In this paper we proposed a new approach to generate flexible parametric families of distributions for modeling survival data. These models arise on CCR scenario, where the latent variables have a zero truncated Poisson distribution. We observed that if $\lambda<0$ ($\lambda>0$) then random variable represents the lifetime of the minimum (maximum) among all elements in risk. Therefore, the extra shape parameter has an important physical interpretation in CCR modeling. 

Moreover, we also proved that depending on the behavior of the baseline hazard function, some of the resulting models will be able to fit data with zero value (instantaneous failures). The parameter estimators are also discussed considering the ML estimation in the presence of random censoring. Furthermore, our proposed methodology is applied in common distributions such as Exponential, Weibull, among other. Many other distributions are also cited, for those our results are valid and may be applied further with success. Finally our proposed methodology is used to describe an real data set related to failure time of devices in an aircraft.

There are a large number of possible extensions of this current work. The presence of covariates as well as long-term survivals are very common in practice \cite{perdona2011}. Our approach should be investigated in these contexts. Another possible approach is to consider bivariate versions using the idea presented by Marshall and Olkin \cite{marshall1997new}.

\section*{Disclosure statement}

No potential conflict of interest was reported by the author(s).

\bibliographystyle{tfs}

\bibliography{reference}

\end{document}